\documentclass[11pt]{article}
\usepackage[margin=1in]{geometry}
\usepackage{titlesec}
\usepackage[utf8]{inputenc}
\usepackage[T1]{fontenc}
\usepackage{graphicx}
\usepackage{grffile}
\usepackage{longtable}
\usepackage{wrapfig}
\usepackage{rotating}
\usepackage[normalem]{ulem}
\usepackage{amsmath}
\usepackage{textcomp}
\usepackage{amssymb}
\usepackage{xcolor}
\usepackage{capt-of}
\usepackage{hyperref}
\usepackage{microtype}
\usepackage{amsmath}
\usepackage{amssymb}
\usepackage{amsthm}
\usepackage{framed}
\usepackage{verbatim}
\usepackage[style=alphabetic,maxbibnames=99,backend=bibtex]{biblatex}
\usepackage{fullpage}
\usepackage[algo2e,linesnumbered]{algorithm2e}
\newtheorem{lemma}{Lemma}[section]
\newtheorem{theorem}[lemma]{Theorem}
\newtheorem{corollary}[lemma]{Corollary}
\newtheorem{definition}[lemma]{Definition}

\newtheorem{question}{Question}
\providecommand{\Pr}{\text{Pr}}
\renewcommand{\Pr}{\text{Pr}}

\DeclareRobustCommand{\BCC}{{\textsc{BCC}}}
\DeclareRobustCommand{\CC}{{\textsc{CC}}}
\DeclareRobustCommand{\conn}{\textsc{Connectivity}\xspace}
\DeclareRobustCommand{\connComps}{\textsc{ConnectedComponents}\xspace}
\DeclareRobustCommand{\partition}{\textsc{Partition}\xspace}
\DeclareRobustCommand{\partitionComp}{\textsc{PartitionComp}\xspace}
\DeclareRobustCommand{\multicycle}{\textsc{MultiCycle}\xspace}
\DeclareRobustCommand{\twocycle}{\textsc{TwoCycle}\xspace}
\DeclareRobustCommand{\twopartition}{\textsc{TwoPartition}\xspace}
\DeclareRobustCommand{\poly}{\text{poly}}

\author{Shreyas Pai \\ The University of Iowa \\ \texttt{shreyas-pai@uiowa.edu} \and Sriram V.~Pemmaraju\\ The University of Iowa \\ \texttt{sriram-pemmaraju@uiowa.edu}}
\date{}
\title{Connectivity Lower Bounds in Broadcast Congested Clique\thanks{A short version of this
paper has appeared as a brief announcement in PODC 2019.}}
\hypersetup{
 pdfauthor={Shreyas Pai and Sriram V.~Pemmaraju},
 pdftitle={Connectivity Lower Bounds in Broadcast Congested Clique},
 pdfkeywords={},
 pdfsubject={},
 pdfcreator={Emacs 25.3.1 (Org mode 9.1.8)},
 pdflang={English}}
\begin{document}

\maketitle

\begin{abstract}
We prove three new lower bounds for graph connectivity in the $1$-bit broadcast congested clique model,
BCC$(1)$.
First, in the KT-$0$ version of BCC$(1)$, in which nodes are aware of neighbors only through port numbers,
we show an $\Omega(\log n)$ round lower bound for \conn even for constant-error randomized Monte Carlo algorithms.
The deterministic version of this result can be obtained via the well-known ``edge-crossing'' argument,
but, the randomized version of this result requires establishing new combinatorial results
regarding the indistinguishability graph induced by inputs.
	In our second result, we show that the $\Omega(\log n)$ lower bound result extends to the KT-$1$ version of
the BCC$(1)$ model, in which nodes are aware of IDs of all neighbors, though our proof works only for deterministic
algorithms. Since nodes know IDs of their neighbors in the KT-$1$ model, it is no longer possible to play
``edge-crossing'' tricks;
instead we present a reduction from the 2-party communication complexity problem \partition
in which Alice and Bob are given two set partitions on $[n]$ and are required to determine if
the join of these two set partitions equals the trivial one-part set partition.
While our KT-$1$ \conn lower bound holds only for deterministic algorithms, in our third result
we extend this $\Omega(\log n)$ KT-1 lower bound to constant-error Monte Carlo algorithms for the closely related
\connComps problem.
We use information-theoretic techniques to obtain this result.
All our results hold for the seemingly easy special case of \conn in which an algorithm
has to distinguish an instance with one cycle from an instance with multiple cycles.
Our results showcase three rather different lower bound techniques and lay the groundwork
for further improvements in lower bounds for \conn in the BCC$(1)$ model.

\end{abstract}
\section{Introduction}
\label{sec:introduction}
We are given an $n$-node, completely connected \textit{communication network} in which
each node can broadcast at most $b$ bits in each round.
These $n$ nodes and a subset of the edges of the communication network form the \textit{input graph}.
The question we ask is this: how many rounds of communication does it take to determine if the input graph
is connected? This is the well known \conn\ problem in the \textit{$b$-bit Broadcast Congested Clique},
i.e., the \BCC$(b)$ model.

A series of recent rapid improvements \cite{hegeman15_towar_optim_bound_conges_clique,ghaffari16_mst_log_star_round_conges_clique,jurdzinski18_mst_o_round_conges_clique} have shown that \conn\ and in fact MST, can be solved in $O(1)$ rounds w.h.p.\footnote{We use ``w.h.p.'' as short for ``with high probability'' which
refers to the probability that is at least $1 - 1/n^c$ for $c \ge 1$.} in the
\textit{$b$-bit Congested Clique} model, \CC$(b)$, when $b = \log n$.
The \CC$(b)$ model allows each node to send a possibly different $b$-bit message to each of the other
$n-1$ nodes in the network, in each round.
In contrast, the fastest known algorithm for \conn\ in the \BCC$(\log n)$ model, due to Jurdzi\'{n}ski and Nowicki \cite{JurdzinskiNowickiDISC2017},
is deterministic and it runs in $O\left(\frac{\log n}{\log\log n}\right)$ rounds.
This contrast between \BCC$(b)$ and \CC$(b)$ is not surprising, given how much larger the overall bandwidth in \CC$(b)$ is compared to \BCC$(b)$.
Becker et al.~\cite{BeckerARRCOCOON2016} show that the \textit{pair-wise set disjointness} problem
can be solved in $O(1)$ rounds in \CC$(1)$, but needs $\Omega(n)$ rounds in \BCC$(1)$.
But, despite the fact that \conn\ is such a fundamental problem, no non-trivial lower bound is known for \conn\ in \BCC$(1)$.
In fact, prior to this paper, we could not even rule out an $O(1)$-round
\conn\ algorithm in BCC$(1)$.

Lower bound arguments in ``congested'' distributed computing models typically use a ``bottleneck''
technique \cite{censor-hillel17_quadr_near_quadr_lower_bound_conges_model,czumaj18_detec_conges,das11_distr_verif_hardn_distr_approx,drucker14_power_conges_clique_model,FischerSPAA2018,HolzerPOPODIS2015}. At a high level, this technique consists of showing that there is a low bandwidth
cut in the communication network across which a high volume of information has to flow in order to solve
the given problem.
The lower bound on information flow is usually obtained via 2-party communication complexity lower bounds
\cite{kushilevitz97_commun_compl}.
Not surprisingly, the ``bottleneck'' technique does not work in the \CC$(b)$ model because any cut with
$\Theta(n)$ vertices in each part, has a high bandwidth of $\Theta(n^2 \cdot b)$ bits. In fact, a result of
Drucker et al.~\cite{drucker14_power_conges_clique_model}, showing that circuits can be simulated efficiently
in the Congested Clique model, indicates that no technique we currently know of can prove non-trivial lower bounds in the \CC$(b)$ model.
However, as further shown by \cite{drucker14_power_conges_clique_model}, ``bottlenecks'' are possible for some problems in the weaker \BCC$(b)$ model.
In this model, every cut has bandwidth $O(n \cdot b)$ and for example Drucker et al.~\cite{drucker14_power_conges_clique_model}
provide a reduction showing that for the problem of detecting the presence of a $K_4$ in the input graph
there is a cut across which $\Omega(n^2)$ information has to flow.
This leads to an $\Omega(n/b)$ lower bound for $K_4$-detection in the \BCC$(b)$.

All known lower bounds \cite{drucker14_power_conges_clique_model,HolzerPOPODIS2015} in the \BCC$(\log n)$ model have this general structure and
these techniques work for problems such as fixed subgraph
detection, all pairs shortest paths, diameter computation, etc., that are relatively
difficult, requiring polynomially many rounds to solve.
For ``simpler'' problems such as \conn\ and MST, we need more fine-grained lower bound techniques
that allow us to prove polylogarithmic lower bounds.
Specifically, since \conn\ can be solved in \BCC$(b)$ for any $b \ge 1$ in just $O(\poly(\log n))$ rounds, the best
we can expect is to show the existence of a cut across which $\Omega(n\cdot \poly(\log n))$ volume
of information needs to flow. In fact, the connected components of a subgraph can be represented
in $O(n \log n)$ bits and this is all that needs to communicated across a cut to solve \conn.
Thus the best lower bound we can expect for \conn\ via this technique is an $\Omega(\log n/b)$.
However, even this was unknown prior to this paper and one contribution of this paper is
an $\Omega(\log n/b)$ lower bound for \conn\ using the ``bottleneck'' technique.

\subsection{Our Contribution}

\noindent
We consider the \conn problem and the closely related \connComps\ problem in the \BCC$(1)$ model.
In the latter problem, each node needs to output the label of the connected component it belongs to.
We work in the \BCC$(1)$ model because it allows us to isolate barriers due to different levels of
initial local knowledge (e.g., knowing IDs of neighbors vs not knowing IDs).
This is also without loss of generality because a $t$-round lower bound in \BCC$(1)$ immediately
translates to a $t/b$-round lower bound in \BCC$(b)$.
We consider two natural versions of the \BCC$(1)$ model, that we call KT-0 and KT-1 (using notation from \cite{awerbuch90_trade_off_between_infor_commun_broad_protoc}).
In the KT-0 (``Knowledge Till 0 hops'') version, nodes are unaware of IDs of other nodes in the network and the $n-1$ communication ports
at each node are arbitrarily numbered 1 through $n-1$.
In the KT-1 (``Knowledge Till 1 hop'') version, nodes know all $n$ IDs in the network and the $n-1$ communication ports at each node
are respectively labeled with the IDs of the nodes at the other end of the port.
Note that if the bandwidth \(b = \Omega(\log n)\), then there is essentially no distinction between the KT-0 and KT-1 versions
since each node in the KT-0 version can send its ID to neighbors in constant rounds and then nodes
would have as much knowledge as they initially do in the KT-1 version.
But the difference in initial knowledge plays a critical role when \(b = o(\log n)\) and in
fact our best results in these two models use completely different techniques.
We present three main lower bound results in this paper, derived using very different techniques.
\begin{itemize}
	\item In the KT-0 version of \BCC$(1)$ we show an $\Omega(\log n)$ round lower bound for \conn\ even for constant-error randomized Monte Carlo
		algorithms.
		In fact, the lower bound is shown for the seemingly simpler ``one cycle vs two cycles'' problem in which the
		input graph is either a single cycle or consists of two disjoint cycles and the algorithm has to distinguish
		between these two possibilities.
		We use a well-known \textit{indistinguishability} argument involving ``edge crossing''
		\cite{KormanKPDC2010,baruch15_random_proof_label_schem,patt-shamir17_proof_label_schem} for this result, but the main novelty here
		is how this argument deals with the possibility that the algorithm can err on a constant fraction of the input
		instances. In a standard edge crossing argument one shows that for a particular YES instance (i.e., a connected or ``one-cycle'' instance) $G$, many of the NO instances $G(e, e')$
		obtained by crossing pairs of edges $e$ and $e'$ in $G$ cannot be distinguished even after some $t$ rounds of a BCC$(1)$ algorithm (see Definition \ref{def:portpreservingcrossing} for the precise definition of a crossing).
		But for a randomized lower bound in \BCC$(1)$, it is not enough to consider a single YES instance.
		Instead, we use the bipartite \textit{indistinguishability graph} induced by all
		YES and NO instances and show that this satisfies a polygamous version of
		Hall's Theorem (see Theorem \ref{thm:poly-halls}). This allows us to show the existence of a large generalized matching
		in the indistinguishability graph, which in turn shows that every \(o(\log n)\) round constant-error Monte Carlo
		algorithm can be fooled into making more errors than it is allowed.

	\item We then show that the above lower bound result extends to the KT-1 version of the BCC$(1)$ model, though our proof only works for
		deterministic algorithms. In KT-1, because of knowledge of IDs of neighbors, it is no longer possible to perform ``edge crossing'' tricks.
		But we are able to successfully use the ``bottleneck'' technique and show that there is
		a cut for the \conn\ problem across which $\Omega(n \log n)$ bits need to flow.
		We prove this result by presenting a reduction from the 2-party communication complexity problem
		\partition \cite{hajnal88_commun_compl_graph_proper}. In the \partition\ problem,
		we have a ground set \([n]\) and Alice and Bob respectively are given two
		set partitions \(P_A\) and \(P_B\) of \([n]\). The goal is to output 1 iff
		\(P_A \vee P_B = \mathbf{1}\) where \(P_A \vee P_B\) (read as ``$P_A$ join $P_B$'') is the
		finest partition \(P\) such that both \(P_A\) and \(P_B\) are refinements of \(P\)
                \footnote{Given two set partitions $P$ and $P'$ of $[n]$, $P$ is said to be a \textit{refinement}
                  of $P'$ if for every part \(S \in P\), there is a part \(S' \in P'\) such that \(S \subseteq S'\). For example the partition \((1,2)(3,4)(5)\) is a refinement of \((1,2)(3,4,5)\).}
                and \(\mathbf{1}\) is the trivial partition consisting of the single set \([n]\).
		For example, if $P_A = (1,2)(3,4)(5)$, $P_B = (1,2,4)(3)(5)$, and $P_C = (1,2,4)(3,5)$ then $P_A \vee P_B = (1,2,3,4)(5)$
		and $P_A \vee P_C = (1,2,3,4,5)$.
		We then use the fact that the deterministic communication complexity of \partition\ is
		$\Omega(n \log n)$ to obtain our result.
		Again, this time using a linear-algebraic argument, we show our result for a seemingly simple special case of \conn: ``one cycle vs
		multiple cycles.'' As far as we know, randomized communication complexity of \partition\ is
		a long-standing unresolved problem.
		Showing a lower bound on the randomized communication complexity of \partition will immediately lead to a KT-1 lower bound for randomized \conn algorithms, via our reduction.

	\item Our final result arises from our attempt to obtain a KT-1 lower bound even for constant-error
		Monte Carlo algorithms.
		We consider a version of the \partition\ problem, called \partitionComp, in which Alice and Bob are required
		to output the join of their respective input partitions $P_A$ and $P_B$
		instead of just determining if $P_A \vee P_B = \mathbf{1}$.
		We use an information-theoretic argument to show that the mutual information of any
		algorithm, even a constant-error Monte Carlo algorithm, that solves this version of
		\partition\ is $\Omega(n \log n)$.
		This leads to an $\Omega(\log n)$-round lower bound for \connComps\
		in the KT-1 version of BCC$(1)$, even for constant-error randomized Monte Carlo algorithms.
\end{itemize}

We prove in this paper the first non-trivial lower bounds for \conn in the \BCC$(1)$ model.  The fact that our
lower bounds hold even in the KT-1 model implies that the difficulty of the problem does not arise
just from lack of knowledge of IDs of other nodes. The fact that our lower bounds hold for extremely
sparse (i.e., 2-regular) graphs, suggests that there might be room to get stronger lower bounds by
considering dense input graphs. In fact, using a deterministic sketching technique
\cite{MontealegreTodincaArxiv2016,MontealegreTodincaPODC2016}, it is possible to obtain a
deterministic $O(\log n)$-round BCC(1) algorithm for \conn for graphs with arboricity bounded by a
constant. This implies that our lower bounds are tight for uniformly sparse graphs.


\subsection{The \BCC\((b)\) Model}
\label{sec:bcc-one-def}
A size-$n$ \textit{\textsc{KT-0} instance} of the \BCC$(1)$ model consists of $n$ vertices, each with a unique
$O(\log n)$-bit ID.
Each vertex has $n-1$ \textit{communication ports} labeled distinctly, 1 through $n-1$,
in an arbitrary manner.
A key feature of the \textsc{KT-0} instance is that port labels have nothing to do with IDs.
Pairs of communication ports are connected by network edges
such that the underlying communication network is a clique.
The $n$ vertices along with a subset of the edges form the \textit{input graph}.
Thus some edges are both network edges and input graph edges, whereas the remaining edges are just
network edges.
The initial knowledge of a vertex $v$ consists of its ID, its port numbering, an identification of ports that correspond to input edges, and an arbitrarily long string $r_v$ of random bits.
In each round $t$, each vertex $u$ receives messages via broadcast from the remaining $n-1$ vertices
in the previous round, performs local computation, and broadcasts a message of length at most $b$-bits.
This message is received at the beginning
of round $t+1$ by the remaining $n-1$ vertices along each of their communication ports that connect
to $u$.
After $t$ rounds, the at most \(t\cdot b\) bits that \(v\) sends and the at most \((n-1)\cdot t\cdot b\) bits that \(v\) receives, along with the ports that they are received from make up the transcript of \(v\) at round \(t\).
A size-$n$ \textit{\textsc{KT-1} instance} of the \BCC$(b)$ model differs from a KT-0 instance in
one important way:
each network edge \(e = \{u, v\}\) is connected to \(u\) at port number \(ID(v)\) and connected to \(v\) at port number \(ID(u)\).
Thus, in a \textsc{KT-1} instance, IDs serve as port numbers and the initial knowledge of a vertex consists
include all $n$ vertex IDs.

Since the main focus of the paper is to derive lower bounds, we assume the \textit{public coin
model} in which all the random strings $r_v$ are identical. Lower bounds proved in the public
coin model hold in the \textit{private coin} model as well, in which all the $r_v$'s are
distinct.
For a decision problem, such as \conn, when we run a \BCC\((b)\) algorithm \(\mathcal{A}\) on an
input graph \(G\), each vertex outputs either YES or NO and the output of the system is YES if all
vertices output YES and is NO otherwise. For a deterministic algorithm \(\mathcal{A}\) for \conn\
the system must output YES if \(G\) is connected and NO if \(G\) is disconnected.
If \(\mathcal{A}\) is an \textit{\(\epsilon\)-error randomized Monte Carlo algorithm}, then in order to be correct,
it must satisfy the following requirements:
(i) if \(G\) is connected then the system outputs YES with probability \(> 1 - \epsilon\) and
(ii) if \(G\) is disconnected then the system outputs NO with probability \(> 1 - \epsilon\).


\subsection{Related Work}
\label{sec:related-work}
\textsc{Congest} model \cite{peleg00_distr_comput} lower bounds via the ``bottleneck technique'' that rely on communication complexity lower bounds have been
shown for MST and related connectivity problems in \cite{das11_distr_verif_hardn_distr_approx} and for
minimum vertex cover, maximum independent set, optimal graph coloring, all pairs shortest paths, and subgraph detection in \cite{censor-hillel17_quadr_near_quadr_lower_bound_conges_model,czumaj18_detec_conges,FischerSPAA2018}.
This approach has also been used to derive \BCC$(\log n)$ lower bounds in \cite{drucker14_power_conges_clique_model,HolzerPOPODIS2015}.
Becker et al.~\cite{BeckerARRCOCOON2016} define a spectrum of congested clique models parameterized by a \textit{range} parameter $r$,
denoting the number of distinct messages
a node can send in a round. Setting $r = 1$ gives us the \BCC$(b)$ model and setting $r = n$ gives us the \CC$(b)$ model.
They show the \textit{pair-wise set disjointness} problem is sensitive to the value of $r$ in the sense that for every pair of ranges $r' < r$,
the problem can be solved provably faster in the model with range $r$ than it can in the model with range $r'$.

Distributed lower bounds via the ``edge crossing'' argument have a long history in distributed computing -- see \cite{KorachMZSICOMP1987} for an example
in the context of proving message complexity lower bounds.
More recent examples \cite{KormanKPDC2010,baruch15_random_proof_label_schem,patt-shamir17_proof_label_schem} appear in the context of \textit{proof-labeling
schemes}.
Informally speaking, a \textit{proof-labeling scheme} consists of a \textit{prover} who labels the vertices of the input configuration
with labels and a \textit{distributed verifier} who is required to verify a predicate (e.g., do the marked edges form an MST?) in one round,
using the help of the prover's labels.
The \textit{verification complexity} of a proof-labeling scheme is the size of the largest message sent by the verifier.
Patt-Shamir and Perry \cite{patt-shamir17_proof_label_schem} show an $\Omega(\log n)$ lower bound on the verification complexity of MST in the broadcast
congested clique model.
An $\Omega(\log n)$ lower bound in the KT-0 version of \BCC$(1)$ for \textit{deterministic} \conn\ algorithms follows
from this result. The high level idea is that if there were a faster \BCC$(1)$ \conn\ algorithm,
the prover could use the transcript of the algorithm at each vertex $v$ as the label at $v$. The verifier could then
broadcast these transcripts and locally, at each vertex $v$, simulate the algorithm at $v$.
Baruch et al.~\cite{baruch15_random_proof_label_schem} show that if there is
a deterministic proof-labeling scheme with verification complexity $\kappa$, then there is a randomized proof-labeling
scheme with one-sided error having verification complexity $O(\log \kappa)$.
Combining this with the fact that MST verification has a deterministic proof-labeling scheme with $O(\log^2 n)$ verification complexity \cite{KormanKPDC2010},
leads to a randomized proof-labeling scheme with $O(\log\log n)$ verification complexity for MST
\cite{baruch15_random_proof_label_schem,patt-shamir17_proof_label_schem}.
This needs to be contrasted with the fact that we show an $\Omega(\log n)$ lower bound for \conn\ in KT-0 \BCC$(1)$ even
for constant-error Monte Carlo algorithms.

There have been recent attempts to combine the edge crossing and bottleneck techniques to obtain lower bounds for triangle detection in the \textsc{Congest} model \cite{abboud17_foolin_views,FischerSPAA2018}. In particular, \cite{FischerSPAA2018} provide an \(\Omega(\log n)\) lower bound for deterministic algorithms solving triangle detection in the KT-1 \textsc{Congest} model with \(1\)-bit bandwidth.

\section{Technical Preliminaries}
\label{section:techPrelim}

\paragraph{Polygamous Hall's Theorem.}
Let \(G = (L, R, E)\) be a bipartite graph. A \textit{\(k\)-matching} is a subgraph consisting of a set of nodes \(A \subseteq L\) where each \(v \in A\) has edges to nodes in the set \(nbr(v)\) such that \(|nbr(v)| = k\) and \(nbr(u) \cap nbr(v) = \emptyset\) for \(u, v \in A\), \(u \neq v\). The size of a \(k\)-matching is the number of connected components in the subgraph.

\begin{theorem}[Polygamous Hall's Theorem]
  \label{thm:poly-halls}
  Let \(G = (L, R, E)\) be a bipartite graph. If for every \(S \subseteq L\) we have \(|N(S)| \ge k |S|\) then \(G\) has a \(k\)-matching of size \(|L|\).
\end{theorem}
\begin{proof}
  Make \(k\) copies of each node in \(L\) while keeping \(R\) the same. Now for every \(S \subseteq L\) we have \(|N(S)| \ge |S|\) and by Hall's marriage theorem, we have a matching in the modified bipartite graph which is a \(k\)-matching of size \(|L|\) in the original graph.
\end{proof}

\paragraph{Yao's Minimax Theorem.} The standard way to prove lower bounds on \(\epsilon\)-error randomized algorithms is by invoking
Yao's Minimax Theorem \cite{yao77_probab}.
Let \(RR_\epsilon(P)\) denote the minimum round complexity of any \(\epsilon\)-error randomized algorithm that solves \(P\).
Let \(DR_\epsilon^\mu(P)\) denote the \textit{distributional round complexity} of \(P\), which is the minimum deterministic
round complexity of an algorithm whose input is drawn from the distribution \(\mu\) (known to the algorithm) and the
algorithm is allowed to make error on at most \(\epsilon\) fraction of the input (weighted by \(\mu\)).

\begin{theorem}
  [Yao's Minimax Theorem]
  \label{thm:yao-minimax-bcc}
  For any problem \(P\), \(RR_\epsilon(P) \ge \max_{\mu}\{DR_\epsilon^\mu(P)\}\)
\end{theorem}
Yao's Minimax Theorem reduces the problem of proving a randomized lower bound to the task of designing a ``hard''
distribution that produces high distributional complexity.

\paragraph{Lower bound for \partition.}
The total number of distinct partitions on a ground set of \(n\) elements is given by the \textit{\(n^{th}\) Bell number} \(B_n\). It is well known that \(B_n = 2^{\Theta(n \log n)}\). This means that the number of
different possible input pairs that Alice and Bob can receive in the \partition\ problem
is \(B_n^2 = 2^{\Theta(n \log n)}\).
Define the matrix \(M^n\) such that \(M^n(i, j) = 1\) if \(P_i \vee P_j = 1\) and \(M^n(i, j) = 0\) otherwise. Note that \(M^n\) is a \(B_n \times B_n\) matrix. Theorem \ref{theorem-dowling-wilson} shows that this matrix is non-singular.

\begin{theorem}
  [\cite{dowling75_whitn_number_inequal_geomet_lattic,welsh10_matroid}]
  \label{theorem-dowling-wilson}
  \(rank(M^n) = B_n\) where \(B_n\) is the \(n^{th}\) Bell number
\end{theorem}

Therefore by Lemma 1.28 of \cite{kushilevitz97_commun_compl} we get the following corollary.
\begin{corollary}
  \label{corollary-det-partition}
  The deterministic 2-party communication complexity of \partition is \(\Omega(n \log n)\)
\end{corollary}

\paragraph{Information Theory.} Let \(\mu\) be a distribution over a finite set \(\Omega\) and let \(X\) be a random variable distributed according to \(\mu\). The \textit{entropy} of \(X\) is defined as \(H(X) = -\sum_{x \in \Omega}{\mu(x) \log \mu(x)}\) and the \textit{conditional entropy} of \(X\) given \(Y\) is \(H(X | Y) = \sum_{y} \Pr[Y = y] H(X | Y = y)\) where \(H(X|Y=y)\) is the entropy of the conditional distribution of \(X\) given the event \(\{Y = y\}\). The \textit{joint entropy} of two random variables \(X\) and \(Y\), denoted by \(H(X, Y)\), is just the entropy of their joint distribution.

The \textit{mutual information} between random variables \(X\) and \(Y\) is \(I(X; Y) = H(X) - H(X|Y) = H(Y) - H(Y|X)\) and the \textit{conditional mutual information} between \(X\) and \(Y\) given \(Z\) is \(I(X; Y | Z) = H(X|Z) - H(X|Y, Z)\). See the first two chapters of \cite{cover06_elemen_infor_theor_wiley_series} for an excellent introduction to the basics of information theory.

\section{Lower Bounds in the \textsc{KT-0} model}
\label{sec:indistinguishability-bccone}
This section is devoted to proving the following theorem.
As mentioned earlier, our lower bound applies to the simpler ``one cycle vs two cycles'' problem
which we will call \twocycle.
In this problem, the input is promised to be either a single cycle or
two disconnected cycles, each of length at least 3 and the goal is to distinguish between these two
types of inputs.

\begin{theorem}
\label{thm:kt0-const-lb}
 For a sufficiently small constant \(0 < \epsilon \le 1/2\), the \(\epsilon\)-error randomized round
complexity of the \twocycle problem in the \BCC$(1)$ \textsc{KT-0} model is bounded below by $\Omega(\log n)$.
\end{theorem}

\noindent
Two KT-0 instances \(I_1\) and \(I_2\) are said to be \textit{indistinguishable}
after \(t\) rounds of an algorithm \(\mathcal{A}\) if the state of each vertex (i.e., the initial
knowledge and the transcript at that vertex) after \(t\) rounds is the same in both the instances.
We first introduce a technical tool called \emph{indistinguishability via port-preserving
crossings}.  This tool has been used to show distributed computing lower bounds in several settings
\cite{KorachMZSICOMP1987,KormanKPDC2010,baruch15_random_proof_label_schem,patt-shamir17_proof_label_schem}
and we heavily borrow notation from \cite{patt-shamir17_proof_label_schem}.
For an edge \(e = (v, u)\) we use the notation \(e(p, q)\) to denote that \(e\) is connected to port
\(p\) at \(v\) and to port \(q\) at \(u\). For this notation to be unambiguous, we must think of the
edge \(e = (v, u)\) as a directed edge \(v \rightarrow u\) even though the graph itself is
undirected.

\begin{definition}
  [Independent Edges \cite{patt-shamir17_proof_label_schem}]
  \label{def:independent-edges}
  Let \(I\) be an instance with input graph \(G = (V, E)\) and let \(e_1 = (v_1, u_1)\) and
  \(e_2 = (v_2, u_2)\) be two edges of \(G\). The edges \(e_1\) and \(e_2\) are said to be
  independent if and only if \(v_1, u_1, v_2, u_2\) are four distinct vertices and
  \((v_1, u_2), (v_2, u_1) \notin E\).  A set of input graph edges is called independent if every
  pair of edges in the set is a pair of independent edges.
\end{definition}

\begin{figure}[t]
  \centering
  \includegraphics[width=0.7\textwidth]{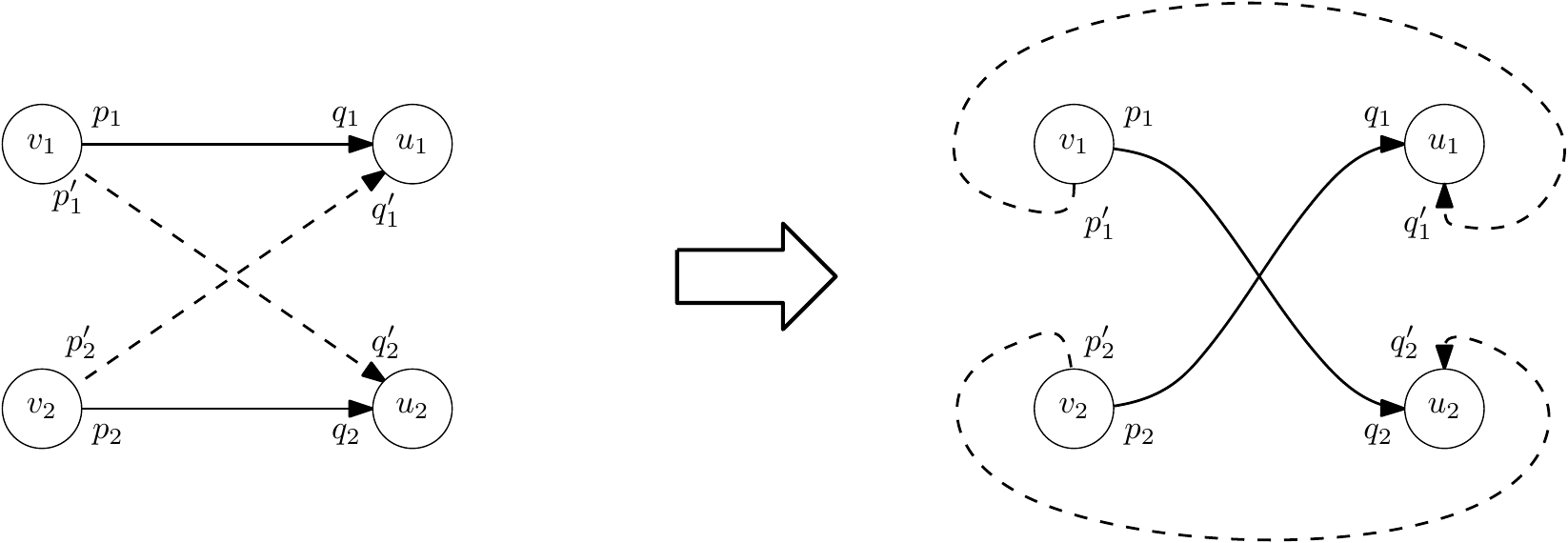}
  \caption{\label{fig:crossing}
	This figure illustrates definition of a port-preserving crossing as per Definition \ref{def:portpreservingcrossing}.
	}
\end{figure}

\vspace{-0.3cm}

\begin{definition}
  [Port-Preserving Crossing \cite{patt-shamir17_proof_label_schem}]
  \label{def:portpreservingcrossing}
  Consider an instance \(I\) with input graph \(G = (V, E)\). Let \(e_1 = (v_1, u_1)\) and
  \(e_2 = (v_2, u_2)\) be two independent edges of \(G\), and let \(e_1' = (v_1, u_2)\) and
  \(e_2' = (v_2, u_1)\) be two corresponding network edges in \(I\).
  Let \(p_1, p_2, q_1, q_2, p_1', q_1', p_2', q_2'\) be eight ports such that
  \(e_1(p_1, q_1), e_2(p_2, q_2), e_1'(p_1', q_2'), e_2'(p_2', q_1')\). The crossing of \(e_1\) and
  \(e_2\) in \(I\), denoted by \(I(e_1, e_2)\), is the instance obtained from \(I\) by replacing
  \(e_1\) and \(e_2\) in \(G\) with the edges \(e_1'\) and \(e_2'\) and rewiring the edges so that
  \(e_1(p_1', q_1'), e_2(p_2', q_2'), e_1'(p_1, q_2),\) and \(e_2'(p_2, q_1)\). (See Figure
\ref{fig:crossing}.)
\end{definition}

\noindent
The following lemma establishes a standard connection between indistinguishability and port-preserving crossings (henceforth ``crossings'') and
is in fact the main motivation for defining crossings. For simplicity, we say that a node sends the character \(\bot\) to denote the fact that the node remains silent. Therefore, the events of a node broadcasting a \(0\), a \(1\), or remaining silent can be described as sending the characters \(0, 1,\) or \(\bot\) respectively.

\begin{lemma}
  \label{lem:crossing-indistinguishable}
  Let \(I\) be an instance with input graph \(G = (V, E)\) and let \(e_1 = (v_1, u_1)\) and
  \(e_2 = (v_2, u_2)\) be two independent edges of \(G\). If \(v_1, v_2\) send the same sequence
  \(x \in \{0, 1, \bot \}^t\) and \(u_1, u_2\) send the same sequence \(y \in \{0, 1, \bot\}^t\) in the first
  \(t\) rounds of the algorithm, then \(I\) is indistinguishable from \(I(e_1, e_2)\) after \(t\) rounds.
\end{lemma}
\begin{proof}
  We will prove the lemma by induction on \(t\). The initial knowledge of each vertex in \(I\) and
  \(I(e_1, e_2)\) is the same so the statement is true for \(t=0\).

  Assume that the lemma is true for some round \(0 \le i \le t\). Therefore, the characters broadcast by the
  vertices in round \(i+1\) will be the same in both the instances. From the definition of port
  preserving crossing it is clear that \(I\) and \(I(e_1, e_2)\) differ only in four edges, \(e_1\),
  \(e_2\), \(e_1' = (v_1, u_2)\), and \(e_2' = (v_2, u_1)\). Therefore, all vertices except
  \(v_1, v_2, u_1\), and \(u_2\) will receive the same characters across all their ports in round \(i+1\)
  in both the instances and hence will have the same state in both instances after round \(i+1\).

  Let the port names of the four edges in \(I\) and \(I(e_1, e_2)\) be as in Definition
  \ref{def:portpreservingcrossing} and Figure \ref{fig:crossing}. In \(I\), the vertex \(u_1\)
  will receive the characters broadcast by \(v_1, v_2\) through ports \(q_1, q_1'\) respectively
  and in \(I(e_1, e_2)\) it will receive the characters broadcast by \(v_2, v_1\) through ports \(q_1,
  q_1'\) respectively. Note that \(v_1\) and \(v_2\) broadcast the same message in round \(i+1\) since
  they send the same sequence \(x\) in the first \(t\) rounds and therefore, the state of \(u_1\)
  after round \(i+1\) will be the same in both instances. We can make similar arguments for
  \(u_2, v_1,\) and \(v_2\) as well. Therefore, the state of each vertex after round \(i+1\) is the
  same in both \(I\) and \(I(e_1, e_2)\) which proves the induction step as well as the lemma.
\end{proof}

\noindent
As a ``warm-up'', we first sketch an easy $\Omega(\log n)$ lower bound for randomized
Monte Carlo algorithms that make \textit{polynomially small error}, i.e., error $\epsilon = 1/n^c$ for constant $c > 0$.
By Yao's minimax theorem (Theorem \ref{thm:yao-minimax-bcc}), it suffices to show a lower bound on the distributional complexity of a
deterministic algorithm under a hard distribution.  Consider the following hard distribution \(\mu\):
Let \(I\) be an arbitrary instance such that the input graph \(G\) of \(I\) is a one-cycle on \(n\)
vertices. Let \(S\) be an arbitrarily chosen set of exactly \(\lfloor n/3 \rfloor\) independent edges
\footnote{Adding an edge to \(S\) invalidates at most two other edges, and therefore we can always
find an independent set \(S\) of size $\lfloor n/3 \rfloor$.} and let \(I(S)\) be the set of all instances
\(I(e, e')\) where \(e, e' \in S\), and therefore, \(|I(S)| = \binom{\lfloor n/3 \rfloor}{2} = \Theta(n^2)\). The
hard distribution \(\mu\) places probability mass \(1/2\) on the instance \(I\) and uniformly
distributes the remaining probability mass among the instances in \(I(S)\).
Now, given a \(t\)-round deterministic algorithm \(\mathcal{A}\) we can assign a \(2t\)-character label to each
edge $(v, u)$ obtained by concatenating the $t$ characters broadcast by $v$ and $u$. Here each character in the label belongs to the alphabet \(\{0, 1, \bot\}\).
The pigeon-hole principle implies that there is a set $S' \subseteq S$, $|S'| \ge n/(3 \cdot 3^{2t})$, of edges in
\(S\) with identical labels.  Then by Lemma \ref{lem:crossing-indistinguishable}, for any $e, e' \in S'$,
\(I\) and \(I(e, e')\) are indistinguishable after \(t\)-rounds of \(\mathcal{A}\).  Since
$\mathcal{A}$ cannot make an error on $I$, it makes errors on all instances $I(e, e')$ where $e, e' \in S'$.
Since $\mu$ assigned the probability mass 1/2 uniformly to all instances in $I(S)$, the probability that
$\mathcal{A}$ makes an error is at least \(|I(S')|/(2|I(S)|) = \binom{|S'|}{2}/\binom{\lfloor n/3 \rfloor}{2} \ge \Omega(3^{-4t})\).
Therefore, if \(t \le 0.001 \cdot c \cdot \log_3 n\), this error becomes \(\Omega(1/n^{0.001c})\)
which is much larger than \(1/n^c\) -- a contradiction, implying that $t > 0.001 \cdot c \cdot \log n$ and
leading to the following theorem.


\begin{theorem}
\label{thm:kt0-whp-lb}
 For any constant \(c > 0\), if \(\epsilon \le 1/n^c\) then the \(\epsilon\)-error randomized round
complexity of the \conn problem in the \BCC$(1)$ \textsc{KT-0} model is \(\Omega(c \cdot \log n)\).
\end{theorem}
\begin{proof}
Note that since the probability mass on \(I\) is so large, any algorithm with permissible error
probability must output YES on \(I\) and therefore, it will also output YES on all instances that
are indistinguishable from \(I\).

Given a \(t\)-round deterministic algorithm \(\mathcal{A}\) we can assign a \(2t\)-character label to each
edge $(v, u)$ where each character belongs to the alphabet \(\{0, 1, \bot\}\). The label is assigned such
that the head $v$ sends the \(i^{th}\) character of the label and the tail $u$ sends
the \((t + i)^{th}\) character of the label in round \(i\) for all edges.  By using the pigeon hole
principle, we see that there is a set $S' \subseteq S$, $|S'| \ge n/(3 \cdot 3^{2t})$, of edges in
\(S\) with identical labels.  By Lemma \ref{lem:crossing-indistinguishable}, for any $e, e' \in S'$,
\(I\) and \(I(e, e')\) are indistinguishable after \(t\)-rounds of \(\mathcal{A}\).  Therefore, any
\(t\) round algorithm will make an error on instances $I(e, e')$ where $e, e' \in S'$ and this makes
the error at least \(\binom{|S'|}{2}/\binom{\lfloor n/3 \rfloor}{2} \ge \Omega(3^{-4t})\).
Therefore, if \(t \le 0.001 \cdot c \cdot \log_3 n\), this error becomes \(\Omega(1/n^{0.001c})\)
which is much larger than \(1/n^c\).
\end{proof}

\noindent
The hard distribution $\mu$ that led to the above theorem fails to give even a super-constant round
lower bound for constant error probability.
This is because for any constant $\epsilon$, there is a constant $t$ such that the error probability $|I(S')|/(2|I(S)|)$
of algorithm $\mathcal{A}$ is smaller than $\epsilon$, leading to no contradiction.


\subsection{A Lower Bound for Constant Error Probability}
\label{sec:crossing-constant}
To get around this problem, we start with the observation that a two-cycle instance $I(e, e')$
obtained from $I$, can also be obtained by crossing edges in other one-cycle instances, i.e., $I(e, e')
= I'(f, f')$ for edges $f, f'$ in an instance $I' \not= I$.
Thus, as the algorithm executes, even though $I(e, e')$ ceases to be indistinguishable from $I$,
it may continue to be indistinguishable from $I'$.
This suggests that we should be considering all one-cycle and two-cycle instances and all
the edge crossings that lead from one-cycle instances to two-cycle instances.
This motivates the definition below of a bipartite \textit{indistinguishability graph} with all
one-cycle and two-cycle instances as vertices.
In the proof of Theorem \ref{thm:kt0-whp-lb}, when we placed the entire probability mass on a single ``star'' indistinguishability graph with $I$
being the central node and instances in $I(S)$ being the leaves, we ran into trouble because the
degree of $I$ in this ``star'' shrank too quickly with the number of rounds, $t$.
If we consider the full indistinguishability graph, we have more leeway. Specifically,
showing the existence of a large matching in the indistinguishability graph would be helpful since
the algorithm is forced to make an error at one of the two endpoints of each matching edge.
We formalize this intuition below, first with some definitions.

Let the set of distinct one-cycle
and two-cycle instances be \(\mathcal{V}_1\) and \(\mathcal{V}_2\) respectively let $\mu$ be a
probability distribution on these.
Let \(\mathcal{A}\) be a \(t\)-round deterministic \textsc{KT-0}
algorithm which solves the \twocycle problem correctly on \((1-\epsilon)\) fraction of input in the
support of \(\mu\) (recall, \(\epsilon\) is a constant).
For any instance \(I \in \mathcal{V}_1 \cup \mathcal{V}_2\), call an edge \(e = (v, u)\) in the
input graph of \(I\) \textit{active} with respect to strings \(x, y \in \{0, 1, \bot\}^t\) iff \(v\)
broadcasts the sequence given by \(x\) and \(u\) broadcasts the sequence given by \(y\) in
the first \(t\) rounds of the algorithm \(\mathcal{A}\). We call an edge active if the strings
\(x, y\) are clear from the context.
\begin{definition}[Indistinguishability Graph]
  Let \(t\) be a non-negative integer and let \(x, y \in \{0, 1, \bot\}^{t}\) be two strings of length \(t\).
  The indistinguishability graph with respect to messages \(x\) and \(y\) after \(t\) rounds of
  algorithm \(\mathcal{A}\) is a bipartite graph
  \(\mathcal{G}^t_{x, y} = (\mathcal{V}_1, \mathcal{V}_2, \mathcal{E}^t)\)
  where \(\mathcal{V}_1\) is the set of all one-cycle instances and \(\mathcal{V}_2\) is the set of
  all two-cycle instances and there is an edge \(\{I_1, I_2\} \in \mathcal{E}^t\) iff \(I_1 \in \mathcal{V}_1\)
  and \(I_2 \in \mathcal{V}_2\) and there exist two active independent directed edges \(e_1 = (v_1, u_1)\)
  and \(e_2 = (v_2, u_2)\) in the input graph of \(I_1\) such that \(I_2 = I_1(e_1, e_2)\).
\end{definition}



\noindent
We now propose to use a rather natural hard distribution \(\mu\) that assigns probability mass \(1/2\) distributed uniformly among the instances
in \(\mathcal{V}_1\) and the remaining probability mass \(1/2\) distributed uniformly among the
instances in \(\mathcal{V}_2\).
We first prove Lemma \ref{lem:degreecondition} that plays a crucial role in our overall proof by essentially
showing that every one-cycle instance has sufficiently many two-cycle neighbors in $\mathcal{G}^t_{x, y}$
with high degree.
This in turn is used in Lemma \ref{lem:hallscondition} to prove that a Polygamous Hall's Theorem (Theorem \ref{thm:poly-halls}) condition holds
for $\mathcal{G}^t_{x, y}$.
This allows us to show that $\mathcal{G}^t_{x, y}$ can be packed with $|\mathcal{V}_1|$ ``stars,''
each with $\Theta(\log n)$ leaves.
We need this generalized notion of a matching because as shown in Lemma \ref{lem:v1-v2-relative-sizes},
\(|\mathcal{V}_2| = |\mathcal{V}_1| \cdot \Theta(\log n)\).
Therefore, the probability mass assigned to an instance in $\mathcal{V}_2$ is $1/\Theta(\log n)$ fraction
of the probability mass assigned to an instance in $\mathcal{V}_1$.
Thus, a ``star'' with its central node from $\mathcal{V}_1$ and $\Theta(\log n)$ leaves from $\mathcal{V}_2$
has roughly equal probability mass assigned to the YES instance and NO instances.

\begin{lemma}
  \label{lem:degreecondition}
  Consider an arbitrary instance \(I_1 \in \mathcal{V}_1\) that is a vertex of \(\mathcal{G}^t_{x, y}\).
  If \(d \ge 1\) is the number of active edges of \(I_1\) with respect to \(x, y\) then for every
  \(i, 3 \le i \le d/2\), \(I_1\) has at least \(d/2\) neighbors of degree \(i \cdot (d-i)\).
\end{lemma}
\begin{proof}
  A two-cycle instance \(I_2 \in \mathcal{V}_2\) will be a neighbor of \(I_1\) iff \(I_1\) and \(I_2\)
  form a pair of crossed instances with respect to \(x, y\). Say \(I_2 = I_1(e, e')\) where \(e = (v, u)\)
  and \(e' = (v', u')\). Note that  \(I_2\) will have two new input graph edges \((v, u')\) and \((u, v')\)
  both of which are active and all input graph edges of \(I_1\) except for \(e, e'\) appear in the
  input graph of \(I_2\). Therefore, \(I_2\) also has \(d\) active edges with respect to
  \(x, y\). The degree of \(I_2\) is determined by the number of active edges either cycle, i.e., if
  \(I_2\) has \(i\) active edges in one cycle and \(d - i\) active edges in the other cycle then its
  degree in \(\mathcal{G}^t_{x, y}\) is \(i \cdot (d-i)\) since we can take one active edge from
  either cycle and cross them to produce a unique neighbor of \(I_2\).

  For every active edge \(e\) in the input graph of \(I_1\), we can associate a unique active edge
  \(e_i\) such that \(I_1(e, e_i)\) has \(i\) active edges in one cycle and \(d - i\) active edges in
  the other cycle. Therefore, \(I_1\) has exactly \(d\) (or \(d/2\) if \(i = d/2\)) neighbors having
  degree \(i(d-i)\). This argument may not hold exactly for \(i = 1, 2\) because \(e\) and \(e_i\) as
  described need not form a pair of independent edges in this case. Thus, the lemma follows.
\end{proof}

\begin{lemma}
  \label{lem:hallscondition}
  For the graph \(\mathcal{G}^t_{x, y}\), consider an arbitrary set \(\mathcal{S} \subseteq
  \mathcal{V}_1\) of one-cycle instances with degree at least \(1\). Let \(N(\mathcal{S})\) be the
  neighborhood of \(\mathcal{S}\) in \(\mathcal{G}^t\). Then \(|N(\mathcal{S})| \ge |\mathcal{S}|
  \cdot \Theta(\log d)\) where \(d\) is the smallest number of active edges in any instance in
  \(\mathcal{S}\).
\end{lemma}
\begin{proof}
  Every \(I \in \mathcal{S}\) has at least \(d\) active edges, therefore by Lemma
  \ref{lem:degreecondition}, there are at least \(d/2\) neighbors of \(I\) having degree \(i \cdot
  (d-i)\) for \(3 \le i \le d/2\). Thus there are at least \((d/2) \cdot |\mathcal{S}| / (i \cdot
  (d-i)) = \Theta(|\mathcal{S}| / i)\) two-cycle instances in \(N(\mathcal{S})\) having degree \(i \cdot
  (d-i)\). Therefore, we have \(|N(\mathcal{S})| \ge \sum_{i = 3}^{d/2}{\Theta(|\mathcal{S}| / i)}\)
  \(= |\mathcal{S}| \cdot \Theta(H_{d/2} - 3/2) \ge |\mathcal{S}| \cdot \Theta(\log d) \), where
  \(H_n\) is the \(n^{th}\) harmonic number.
\end{proof}

\begin{lemma}
  \label{lem:v1-v2-relative-sizes}
  \(|\mathcal{V}_2| = |\mathcal{V}_1| \cdot \Theta(\log n)\).
\end{lemma}
\begin{proof}
  Let \(\mathcal{G} = \mathcal{G}^0_{\lambda, \lambda}\) (\(\lambda\) is the empty string) be the
  indistinguishability graph at round \(0\). Note that in \(\mathcal{G}\), every instance in
  \(\mathcal{V}_1 \cup \mathcal{V}_2\) has strictly positive degree since each instance has \(n\)
  active edges. Therefore, we have \(|\mathcal{V}_1| = |N(\mathcal{V}_2)|\) and \(|\mathcal{V}_2| =
  |N(\mathcal{V}_1)|\). Therefore, by Lemma \ref{lem:hallscondition}, we have \(|\mathcal{V}_2| =
  |\mathcal{V}_1| \cdot \Omega(\log n)\). Now we show that \(|\mathcal{V}_2| = |\mathcal{V}_1| \cdot
  O(\log n)\).

  Since each instance has \(n\) active edges, each one-cycle instance \(I_1\) has degree
  \(n(n-3)/2\) because for each input graph edge \(e\) of \(I_1\) there are \((n-3)\) active edges
  independent of \(e\), which we can cross with to get a unique neighbor of \(I_1\). We need to divide
  by a factor of two because \(I_1(e, e') = I_1(e', e)\). And each two-cycle instance \(I_2\) with the
  smaller cycle having length \(i\) has degree \(i \cdot (n-i)\) since we can cross any two edges in
  different cycles to get a neighbor of \(I_2\).

  Let \(\mathcal{T}_i\) denote the set of two-cycle instances with the smaller cycle having length
  \(i\) for \(3 \le i \le n/2\).

  For every input graph edge \(e\) in a one-cycle instance \(I\), there is exactly one input graph
  edge \(e_i\) such that \(I(e, e_i) \in \mathcal{T}_i\). Therefore, for \(3 \le i < n/2\), each one
  cycle instance has \(n\) neighbors such that the smaller cycle is of length \(i\). And if \(n\) is
  even, each one-cycle instance will have \(n/2\) neighbors where both cycles have length \(n/2\)
  instead.

  We will now show that \(|\mathcal{T}_i| \le |\mathcal{V}_1| \cdot n/(i \cdot (n-i))\). To see this
  note that if we restrict our attention to the subgraph of \(\mathcal{G}\) spanned by instances in
  \(\mathcal{V}_1 \cup \mathcal{T}_i\) then we have a bipartite graph where each instance in
  \(\mathcal{V}_1\) has the same degree \(n\) (or \(n/2\) if \(i = n/2\)) and each instance in
  \(\mathcal{T}_i\) has the same degree \(i \cdot (n-i)\). Therefore, the total number of edges
  incident on \(\mathcal{V}_1\) is \(\le |\mathcal{V}_1| \cdot n\) and those incident on
  \(\mathcal{T}_i\) is \(|T_i| \cdot i \cdot (n-i)\). Since the number of edges should be the same
  counted from either side, we get \(|\mathcal{T}_i| \le |\mathcal{V}_1| \cdot n/(i \cdot
  (n-i))\). Now we finish the proof of the lemma with the following calculation:
  \[|\mathcal{V}_2| = \sum_{i = 3}^{n/2}{|\mathcal{T}_i|} \le \sum_{i}{\frac{n}{i \cdot (n-i)} \cdot |\mathcal{V}_1|} = |\mathcal{V}_1| \cdot O(\log n)\]
\end{proof}

\begin{proof}(of Theorem \ref{thm:kt0-const-lb})
Consider an arbitrary one-cycle instance \(I_1 \in \mathcal{V}_1\) after \(t = 0.1 \log_3 n\) rounds
of algorithm \(\mathcal{A}\). Let \(x, y \in \{0, 1, \bot\}^t\) be the strings that correspond to the largest
set of active edges after \(t\)-rounds of algorithm \(\mathcal{A}\). We would like to count the size
of this set of active edges. Recall that we orient each input graph edge of \(I_1\) in a clockwise
direction. Therefore, each input graph edge in \(I_1\) can be labeled with a string of length \(2t\) which
denotes messages sent across it from the head and the tail (in order) across the \(t\) rounds. This
means that there are at least \(n/3^{2t} = n^{0.8}\) input graph edges in \(I_1\) that have the same
messages sent across them. Therefore, the size of the set of active edges with respect to \(x, y\)
is at least \(\Omega(n^{0.8})\).

By Lemma \ref{lem:hallscondition} and Theorem \ref{thm:poly-halls}, we can say that there exists a
\(\Theta(\log n)\)-matching in \(\mathcal{G}^t_{x, y}\) of size \(|\mathcal{V}_1|\).
No matter what the algorithm \(\mathcal{A}\) outputs on any one-cycle instance, it will produce the
same output on the matched \(O(\log n)\) two-cycle instances. By Lemma
\ref{lem:v1-v2-relative-sizes}, we know that for any \(I_1 \in \mathcal{V}_1\) and
\(I_2 \in \mathcal{V}_2\), \(\mu(I_1) = \mu(I_2) \cdot \Theta(\log n)\) Therefore, each instance
\(I_1 \in \mathcal{V}_1\) contributes to \(\Theta(\mu(I_1))\) the error of the algorithm which means that any
\(t\)-round \BCC\((1)\) algorithm will have total error at least a constant. This implies the
theorem.
\end{proof}
\section{Lower Bounds in the \textsc{KT-1} Model}
\label{sec:partition-deterministic}
Our lower bounds in the KT-1 model are inspired by the work of Hajnal et
al.~\cite{hajnal88_commun_compl_graph_proper}, which is concerned with 2-party
communication complexity of several graph problems, including \conn. In their setup
\cite{hajnal88_commun_compl_graph_proper}, the input graph \(G = (V, E)\) is
\textit{edge-partitioned} among Alice and Bob in such a way that both parties know $V$ and Alice and
Bob respectively know edge sets $E_A$ and $E_B$, were $(E_A, E_B)$ forms a partition of $E$.
One simple deterministic protocol that solves \conn\ in this setup is this: Alice sends all the
connected components induced by \(E_A\) to Bob, who can determine if \(G\) is connected. The worst
case communication complexity of this protocol is \(O(n \log n)\).  Via reduction from \partition,
Hajnal et al. \cite{hajnal88_commun_compl_graph_proper} show that there exists a family of input
graphs such that for \emph{any equal sized} edge partition, the communication complexity of \conn\
is \(\Omega(n \log n)\).

It does not seem possible to reduce from this edge-partitioned version of 2-party \conn\ to \conn\
in the KT-1 model because KT-1 algorithms are vertex-centric and Alice and Bob may not hold all the
edges they need to simulate vertices executing a KT-1 algorithm.  We resolve this issue by designing
a new reduction, from \partition\ to a vertex-partition version of 2-party \conn.  In the Hajnal et
al.~\cite{hajnal88_commun_compl_graph_proper} reduction, \partition\ is reduced to \conn\ on a
family of dense graphs.  Motivated by our KT-0 lower bound for \conn\ for the \twocycle problem, we are interested in deriving a KT-1 \conn\ lower bound for a \textit{sparse}
class of graphs as well.
In what follows, we extend the reduction of Hajnal et al.~from \partition\ to \conn\ in two
important ways: (i) we reduce to a vertex-partitioned version of \conn\ and (ii) we reduce to a
sparse special case of \conn\ that we call the \multicycle\ problem, in which the input is
either a single cycle or two or more cycles, each having length at least \(4\).

\subsection{A Special Case of the \partition\ Problem}
\label{sec:partition-problem}

In order to establish a lower bound for \multicycle, we now consider a special case of the 2-party
\partition\ problem, which we call \twopartition. The input to \twopartition\ consists of partitions
$P_A$ and $P_B$ of \([n]\), for even $n$, such that each part in $P_A$ and $P_B$ has
exactly two elements in it.  We will now use a linear algebraic argument to show that there is an
$\Omega(n \log n)$ deterministic lower bound on this special case of \partition\ also.  The $0$-$1$
matrix \(E^n\) associated with this problem is a sub-matrix of the matrix \(M^n\) where \(M^n(i, j)
= 1\) if \(P_i \vee P_j = 1\) and \(M^n(i, j) = 0\) otherwise (see Section \ref{section:techPrelim}).
The matrix \(E^n\) has dimension \(r \times r\) where $r = n!/(2^{n/2} \cdot (n/2)!)$. This fact
follows from a simple counting argument. In the following theorem, we show that this sub-matrix
$E^n$ has full rank.

\begin{lemma}
  \label{lem:rank-submatrix}
  \(rank(E^n) =  r\) where \(r = n!/(2^{n/2} \cdot (n/2)!)\).
\end{lemma}
\begin{proof}
  We will prove a more general observation -- every sub-matrix \(A_S\) of a full rank \(d \times d\)
  matrix \(A\) formed by choosing a subset \(S\) of the rows and the corresponding columns has rank
  \(s\) where \(s = |S|\). In other words, for all \(S\), \(A_S\) is a full rank \(s \times s\)
  matrix.

  Let \(B\) be a \(d \times d\) diagonal matrix where \(B(i, i) = 1\) if \(i \in S\) and \(B(i, i) = 0\)
  if \(i \notin S\). It is easy to see that \(rank(B) = |S| = s\). Using basic properties of rank,
  \(rank(AB) \le rank(B) \le s\) and by Sylvester's rank inequality \footnote{For any two \(n \times
    n\) matrices \(A, B\), \(rank(AB) \ge rank(A) + rank(B) - n\). We can prove this inequality by
    applying the rank-nullity theorem to the inequality \(null(AB) \le null(A) + null(B)\).},
  \(rank(AB) \ge rank(A) + rank(B) - d = d + s - d = s\).

  Therefore, \(rank(AB) = s\) which means that some minor of \(AB\) having dimension \(s\) needs to
  be of full rank. The only such candidate is the minor corresponding to the matrix \(A_S\) because
  all other minors of dimension \(s\) either have an all zero row or all zero column. Therefore,
  \(A_S\) has full rank.

  Now \(E^n\) is a submatrix of \(M^n\) where the rows and columns correspond to partitions of
	\([n]\) such that each part has exactly two elements in it. Therefore, the lemma follows
  since \(M^n\) has full rank.
\end{proof}

\noindent
By using Stirling's approximation, it can be verified that $r = 2^{\Theta(n \log n)}$.  Then, by the
rank bound and Lemma 1.28 of \cite{kushilevitz97_commun_compl} we get the following corollary.

\begin{corollary}
  \label{corollary-det-eqpartition}
  The deterministic 2-party communication complexity of \twopartition is \(\Omega(n \log n)\)
\end{corollary}

We describe our reductions in the next two subsections. In section \ref{sec:partition-to-2-party-conn}, we reduce the \partition (\twopartition) problem to the vertex partitioned 2-party \conn (2-party \multicycle) problem and in section \ref{sec:bcc-simulation}, we reduce the 2-party \conn (2-party \multicycle) problem to \conn\ (\multicycle) in the KT-1 model.

\subsection{Reductions from \partition\ and \twopartition}
\label{sec:partition-to-2-party-conn}
Here we present two reductions, first from \partition\ to 2-party \conn\ and next from
\twopartition\ to 2-party \multicycle. Alice is given a partition \(P_A = (S_1, S_2, \dots, S_n)\) over the ground set \([n]\)
where \(S_i\) is the \(i^{th}\) part of \(P_A\), which could possibly be
empty if \(P_A\) has fewer than \(i\) parts. Similarly, Bob is given a partition
\(P_B = (S_1', S_2', \dots, S_n')\).
They construct a graph \(G(P_A, P_B)\) as follows: Alice creates vertex sets \(A = \{a_1, \dots, a_{n}\}\)
and \(L = \{\ell_1, \dots, \ell_{n}\}\) whereas Bob creates the vertex sets \(R = \{r_1, \dots, r_{n}\}\)
and \(B = \{b_1, \dots, b_{n}\}\).  Alice and Bob add edges \((\ell_i, r_i)\) for \(i \in [n]\), independent of \(P_A\) and \(P_B\). Alice adds edges between \(A\) and \(L\)
that induce the partition \(P_A\) on \(L\). That is, for every \(S_i \in P_A\), Alice adds edges
\((a_i, \ell_j)\) for all \(j \in S_i\).
There will be some vertices in \(A\) that are not connected to any vertex,
so Alice just adds an edge between these vertices and an arbitrary vertex \(\ell_* \in L\). Bob
similarly adds edges between the sets \(B\) and \(R\). See Figure \ref{fig:partition-reduction}.

\begin{figure}[t]
\centering
\includegraphics[width=.8\linewidth]{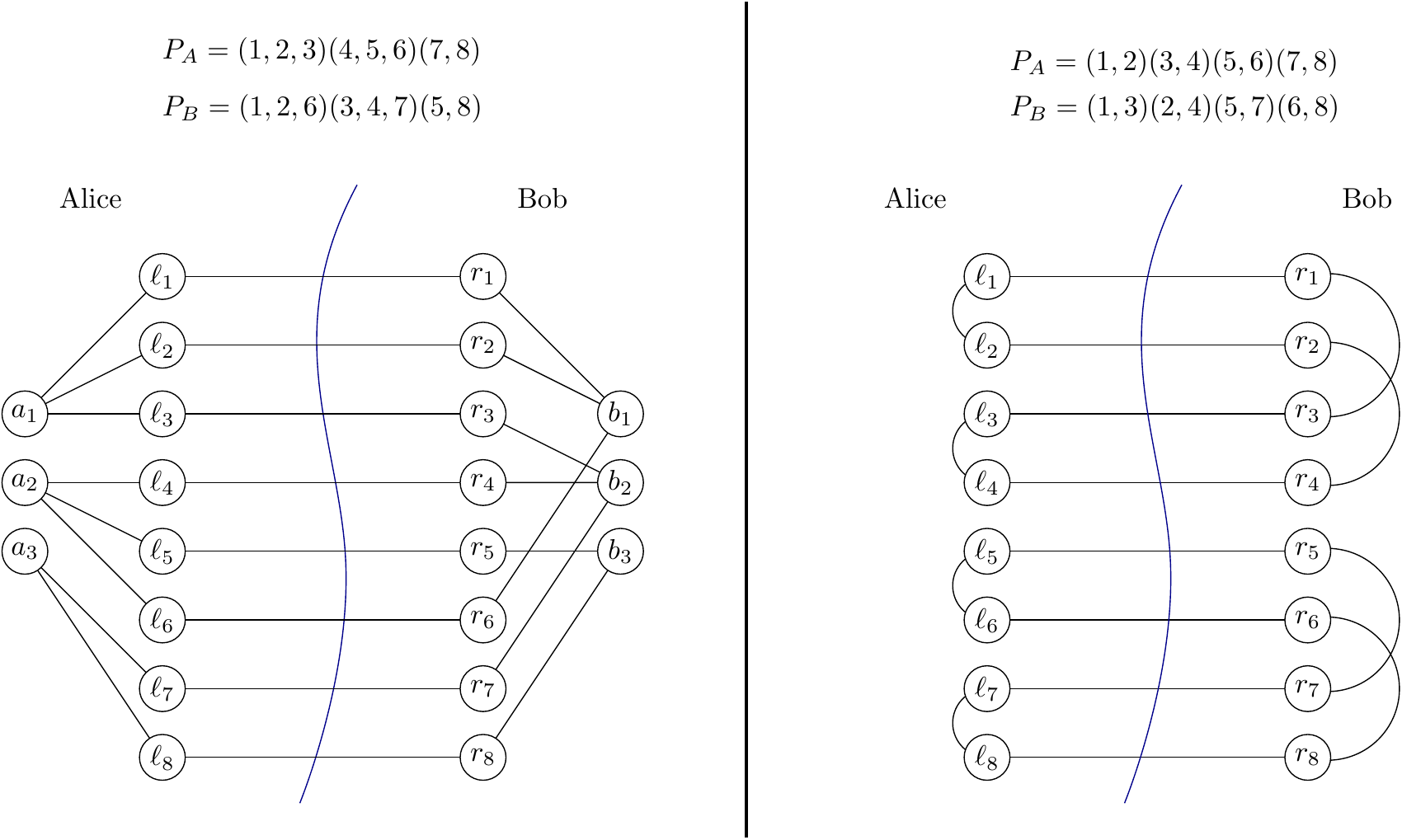}
\caption{\label{fig:partition-reduction} The figure on the left illustrates the reduction from
\partition to 2-party \conn and the figure on the right illustrates the reduction from \twopartition
to 2-party \multicycle.
The vertices \(a_4, \dots,
a_8\) that are connected to \(\ell_* = \ell_8\) and \(b_4, \dots, b_8\) connected to \(r_* = r_8\)
are not shown in the left figure.}
\end{figure}

If \(P_A\) and \(P_B\) are instances of \twopartition, that is, each part of \(P_A\) and
\(P_B\) is of size exactly two, then we can modify the construction of \(G(P_A, P_B)\) by getting
rid of the sets \(A\) and \(B\). Note that in this case \(P_A = (S_1, S_2, \dots, S_{n/2})\) and
\(P_B = (S_1', S_2', \dots, S_{n/2}')\) where each \(S_i\) and \(S_i'\) has size exactly two. If
\(\{i, j\} \in P_A\) then Alice creates an edge between \(\ell_i\) and \(\ell_j\) and Bob does the
same with \(R\) for every pair in \(P_B\). With this modified construction, each vertex in \(G(P_A, P_B)\)
has degree exactly \(2\) and therefore, every connected component of \(G(P_A, P_B)\) will be
a cycle. See Figure \ref{fig:partition-reduction}.

The following theorem encapsulates a crucial property of the graph \(G(P_A, P_B)\) which implies
the correctness of our reductions.
\begin{theorem}
  \label{thm:partition-to-2p-connectivity}
  If \(P_A\) and \(P_B\) are instances of \partition (or \twopartition), then the partition induced
  by the connected components of \(G(P_A, P_B)\) on the vertices in \(L\) and \(R\) corresponds to the
  partition \(P_A \vee P_B\).
\end{theorem}
\begin{proof}
  Call two elements \(a\) and \(b\) \textit{reachable} from each other if there exists a sequence of
  distinct elements \(e_0, e_1, \dots e_t, 1 \le t \le n\) such that \(e_0 = a\), \(e_t = b\) and each
  pair \((e_i,e_{i+1})\) either belongs to the same part of \(P_A\) or the same part of \(P_B\). Any
  partition in which all reachable elements are in the same part have both \(P_A\) and \(P_B\) as
  refinements.

  We claim that two elements belong to the same part of \(P_A \vee P_B\) if and only if they are
  reachable from each other. The backward direction is true because \(P_A\) and \(P_B\) are both
  refinements of \(P_A \vee P_B\). The forward direction is true because if \(a\) and \(b\) are not
  reachable from each other but still belong to the same part \(S\) of \(P_A \vee P_B\) then we can
  refine the part \(S\) to be \(S_a, S_b\) where \(S_a\) is the set of all elements in \(S\) that are
  reachable from \(a\) and \(S_b\) is the set of all elements in \(S\) that are reachable from
  \(b\). It is easy to see that \(S_a\) and \(S_b\) are disjoint. Let \(P'\) be the partition \(P_A \vee P_B\)
  where \(S\) is further refined to be \(S_a, S_b, S \setminus (S_a \cup S_b)\). Note that both \(P_A\)
  and \(P_B\) still remain refinements of the \(P'\) which contradicts the minimality of the join.

  The theorem follows by observing that \(i\) and \(j\) are reachable from each other if and only if
  there is a path from \(\ell_i\) to \(\ell_j\) (and consequently from \(r_i\) to \(r_j\)) in \(G(P_A,P_B)\).
\end{proof}

\subsection{Reductions from 2-party \conn\ and \multicycle}
\label{sec:bcc-simulation}

We now show reductions from 2-party \conn\ to \conn\ in the KT-1 model and from 2-party
\multicycle\ to \multicycle\ in the KT-1 model.  Given an \(r\)-round \textsc{KT-1} algorithm
\(\mathcal{A}\), Alice and Bob will simulate the algorithm with \(G(P_A, P_B)\)
as the input graph. Alice hosts vertices in \(A \cup L\) and Bob hosts vertices in
\(B \cup R\).  For $1 \le i \le n$, the IDs of vertices \(a_i\), \(\ell_i\), \(r_i\), and \(b_i\) are \(i\), \(n + i\), \(2n + i\), and \(3n + i\) respectively. So both
parties know the ID's of all vertices as well as the ID's of neighbors of all hosted vertices in \(G(P_A, P_B)\) and hence, the initial knowledge of hosted vertices.

In order to simulate round \(t\) of \(\mathcal{A}\), Alice
and Bob need to compute the states of all hosted vertices after round \(t\) of
\(\mathcal{A}\). The state of a vertex \(v\) after round \(t\) depends on the initial knowledge and the transcript \(\tau(v, t)\) of \(v\). Assume that Alice and Bob know the
states of all the vertices they host after round \(t-1\). Alice and Bob send a message from \(\{0, 1, \bot\}^{2n}\) to each other. These messages denote the characters their hosted vertices broadcast in round \(t\), in increasing order of
ID. Therefore, they know the sender ID of a character from the position of the character in the
message. This enables Alice and Bob to compute the transcript \(\tau(v, t)\) and hence the state after round \(t\) of all hosted vertices \(v\).

Therefore, in simulating each round, Alice and Bob exchange exactly \(O(n)\) bits with each other and
the total communication complexity of the protocol is \(O(rn)\). If \(\mathcal{A}\) solves the \conn
or \multicycle problems, then using corollaries \ref{corollary-det-partition} and
\ref{corollary-det-eqpartition} respectively and Theorem \ref{thm:partition-to-2p-connectivity}, we obtain the
following result.
\begin{theorem}
  \label{thm:kt1-det-lb}
  The round complexity of a deterministic algorithm for solving the \conn and \multicycle problems in the \textsc{KT-1} model is \(\Omega(\log n)\).
\end{theorem}

\subsection{Information-theoretic Lower Bound for \connComps}
J\'{a} J\'{a} \cite{JaJaJACM1984} proves a lower bound for 2-party \connComps\ and points out
that his techniques may not work for decision problems, indicating that it might be easier to prove
lower bounds for \connComps.  This motivates us to consider the
\connComps\ problem as a lower bound candidate, closely related to \conn, but for which we may
be able to prove an $\Omega(\log n)$ lower bound in the KT-1 model, \textit{even for constant-error
Monte Carlo algorithms}.  It turns out that we are able to prove this result by combining the reductions
described in the previous section with information-theoretic techniques.  We first define the
2-party problem \partitionComp\ which is closely related to \partition, but requires an
output with a large representation.  As in \partition, Alice and Bob are respectively
given set partitions $P_A$ and $P_B$ of $[n]$ and at the end of the communication
protocol for \partitionComp, Alice and Bob are required to output the join $P_A \vee P_B$.  From
Theorem \ref{thm:partition-to-2p-connectivity}, we get that if there is a $t$-round, $\epsilon$-error
Monte Carlo algorithm $\mathcal{A}$ for \connComps\ in the KT-1 model, then there is an
$\epsilon$-error Monte Carlo protocol that solves \partitionComp\ with communication complexity $t
\cdot n$.

Consider the following distribution over inputs of \textsc{PartitionComp}: Alice's input \(P_A\) is
chosen uniformly at random from the set of all partitions and Bob's partition is fixed to be the
finest partition, i.e., \(P_B = (1)(2)(3)\dots(n)\).  With $P_B$ fixed in this manner, $P_A \vee P_B
= P_A$ and at the end of the protocol Bob learns $P_A$. Since $P_A$ is chosen from the uniform distribution,
it's initial entropy is high -- $\Theta(n \log n)$ since the support of the distribution has size $2^{\Theta(n \log n)}$.
Therefore Bob will learn a lot of information by the end of the protocol.
This idea is formalized in the proof of the following theorem.
This proof also has to deal with the complication that the protocol has constant error probability.

\begin{theorem}
  \label{thm:kt1-conn-comp-lb}
  For any constant \(0 < \epsilon < 1\), the round complexity of an \(\epsilon\)-error randomized Monte Carlo algorithm that solves the \connComps problem in the \textsc{KT-1} version of the \BCC\((1)\) model is $\Omega(\log n)$.
\end{theorem}
\begin{proof}
Using Yao's minimax theorem (Theorem
\ref{thm:yao-minimax-bcc}) we can assume that all protocols are deterministic but are allowed to
make an error on \(\epsilon\)-fraction of the input, weighted by \(\mu\). Although appealing to
Yao's theorem is not necessary, it allows us to simplify the exposition. Let \(\Pi\)
denote the transcript of a 2-party protocol that solves \textsc{PartitionComp} and let \(|\Pi|\) denote the length of the longest transcript produced by \(\Pi\) on any input. We know that
\[|\Pi| \ge H(\Pi(P_A, P_B)) \ge I(\Pi(P_A, P_B); P_A, P_B) = I(P_A, P_B; \Pi(P_A, P_B)) = I(P_A; \Pi(P_A, P_B))\]

\noindent
where the last equality follows from the fact that \(P_B\) is fixed according to \(\mu\). From the
definition of mutual information, \(I(P_A; \Pi(P_A, P_B)) = H(P_A) - H(P_A| \Pi(P_A, P_B))\).
Alice's input \(P_A\) is uniformly distributed among all $B_n = 2^{\Theta(n \log n)}$ set partitions according to the hard
distribution \(\mu\). Therefore \(H(P_A) = \Theta(n \log n)\). Let \(B\) be the set of protocol transcripts
that produce an error on the input \(P_A, P_B\). If \(\Pi(P_A, P_B) \notin B\) then \(H(P_A|\Pi(P_A, P_B)) = 0\)
since the output of the protocol is \(P_A \vee P_B = P_A\). We are guaranteed
that \(\Pr[\Pi(P_A, P_B) \in B] \le \epsilon\). Therefore, the second term can be bounded as
follows.
\begin{align*}
  H(P_A| \Pi(P_A, P_B)) &= \sum_{\pi} \Pr[\Pi(P_A, P_B) = \pi] H(P_A | \Pi(P_A, P_B) = \pi) \\
                        &= \sum_{\pi \in B}\Pr[\Pi(P_A, P_B) = \pi] H(P_A | \Pi(P_A, P_B) = \pi) \le \epsilon H(P_A)
\end{align*}

\noindent
Where the last inequality follows from the fact that \(H(X|Y) \le H(X)\) for any \(X, Y\). This
implies \(I(P_A; \Pi(P_A, P_B)) = \Omega(n \log n)\) which proves that any \(\epsilon\)-error
randomized protocol that solves the \textsc{PartitionComp} problem has communication complexity of
\(\Omega(n \log n)\). This in turn implies that \(t = \Omega(\log n)\) which proves the theorem.
\end{proof}

\section{Future Work}
\label{sec:future-work}
In this paper, we used various techniques to obtain better lower bounds for \conn in the \BCC\((1)\) model.
However, these bounds are still quite weak and the gap between these lower bounds and the best upper bound
is substantial.
The fundamental question that motivated this paper, one that is still open is this.

\begin{question}
  Can we obtain \(\omega(\log n)\) round lower bounds for \conn in the \BCC\((1)\) model or show that this is not possible by
	designing an algorithm running in $O(\log n)$ rounds?
\end{question}

Another way to ask this question is can we obtain super-constant round lower bounds in the \BCC\((\log n)\) model?
It is worth noting again that we have a deterministic upper bound for \conn of \(O(\log n / \log \log n)\) \cite{jurdzinski18_mst_o_round_conges_clique}
in \BCC$(\log n)$, whereas our results do not imply a better than $\Omega(1)$ lower bound in \BCC$(\log n)$.

A second open question, one that is more relevant to the techniques used in this paper is the following.

\begin{question}
  Can we get an \(\Omega(n \log n)\) lower bound on the randomized constant-error communication complexity for
	the \partition and \twopartition problems?
\end{question}

Using the reductions in this paper, a positive answer to this question would imply an $\Omega(\log n)$ lower bound
for \conn\ in \BCC$(1)$ KT-1 model even for constant-error randomized algorithms.

\printbibliography

\end{document}